\documentclass[letterpaper, 10 pt, conference]{ieeeconf}  

\IEEEoverridecommandlockouts                              
 \usepackage[utf8]{inputenc}
 
\usepackage{algorithmic}
\usepackage[dvipdfmx]{graphicx} 
\usepackage[dvipdfmx]{color} 
\usepackage[svgnames]{xcolor} 
\usepackage{textcomp}
\usepackage{comment}

\usepackage{mathtools}
\usepackage{bm}
\usepackage{amsmath,amssymb,amsfonts}

\usepackage{ulem}

\newtheorem{thm}{Theorem}

\newcommand{\INS}[1]{{\color{black}{#1}}}

\title{\LARGE 
\textbf{Data-driven phase control for limit-cycle oscillators \\ under partial observation}
}

\author{Koichiro Yawata, Norihisa Namura, Yuzuru Kato, and Hiroya Nakao
\thanks{K. Yawata is with the Department of Systems and Control Engineering, 
Institute of Science Tokyo, Tokyo 152-8552, Japan
        {\tt\small koichiro.yawata.rt@gmail.com}}%
\thanks{N. Namura is with the Department of Systems and Control Engineering, 
Institute of Science Tokyo, Tokyo 152-8552, Japan
        {\tt\small namura.n.60d9@m.isct.ac.jp (namura.n.aa@m.titech.ac.jp)}}%
\thanks{Y. Kato is with the Department of Complex and Intelligent Systems, Future University Hakodate, Hokkaido 041-8655, Japan
        {\tt\small katoyuzu@fun.ac.jp}}%
\thanks{H. Nakao is with the Department of Systems and Control Engineering
and Research Center for Autonomous Systems Materialogy, 
Institute of Science Tokyo, Tokyo 152-8552, Japan
        {\tt\small nakao.h.ee74@m.isct.ac.jp (nakao@sc.e.titech.ac.jp)}}%
\thanks{
H.N. acknowledges the support from JSPS KAKENHI 
\INS{(grant numbers: 25H01468, 25K03081, and 22H00516).}
\INS{N.N. acknowledges the financial support from JSPS KAKENHI (No. JP25KJ1270).}
Y.K. acknowledges JSPS KAKENHI JP22K14274 and JST PRESTO JP-MJPR24K3.
}
}

\begin{document}

\maketitle
\thispagestyle{empty}
\pagestyle{empty}

\begin{abstract}

Controlling rhythmic systems, typically modeled as limit-cycle oscillators, is an important subject in real-world problems.
Phase reduction theory, which simplifies the multidimensional oscillator state under weak input to a single phase variable, is useful for analyzing the oscillator dynamics.
In the control of limit-cycle oscillators with unknown dynamics, the oscillator phase should be estimated from time series under partial observation in real time.
In this study, we present an autoencoder-based method for estimating the oscillator phase using delay embedding of observed state variables.
We evaluate the order of the phase estimation error under weak inputs and apply the method to phase-reduction-based feedback control of mutual synchronization of two oscillators under partial observation.
The effectiveness of our method is illustrated by numerical examples using two types of limit-cycle oscillators, the Stuart-Landau and Hodgkin-Huxley models.
\end{abstract}

\section{Introduction}\label{sec:intro}

Spontaneous rhythmic phenomena are widely observed in the real world, and their mechanisms and applications have garnered significant interest~\cite{Winfree1980,Pikovsky2001,Kuramoto2003,Kralemann2013,Stankovski2015}.
Such rhythmic phenomena can be mathematically modeled as limit-cycle oscillators~\cite{Winfree1980, Kuramoto2003, nakao2016phase,Ermentrout2010,monga2019phase,ermentrout2019recent}.
Synchronization of limit-cycle oscillators,
which refers to alignment of their dynamics either by mutual interaction or by unilateral forcing,
is one of the most common phenomena in rhythmic systems.
Control of synchronization is a fundamental research subject in various fields of science and engineering, for example, synchronization in power grids~\cite{Dorfler2012synchronization},
gait transitions in hexapod robots~\cite{Namura2025central},
and suppression of pulsus alternans in the heart~\cite{wilson2017spatiotemporal}. 

The phase reduction theory~\cite{Winfree1980, Ermentrout2010, glass1988clocks,nakao2016phase,shirasaka2020phase,wilson2022adaptive}
is useful for analyzing oscillator dynamics under weak inputs 
because it systematically reduces the multidimensional state of the oscillator to a single phase variable that increases with a constant natural frequency along the limit cycle and in its basin of attraction, called the asymptotic phase.
Phase control of limit-cycle oscillators based on phase reduction theory (and 
recently extended phase-amplitude reduction theory~\cite{mauroy2013isostables, wilson2016isostable, shirasaka2020phase}) has been extensively studied,
for example, optimal phase control~\cite{moehlis2006optimal, Monga2019optimal, namura2024optimal2},
optimizing periodic inputs~\cite{Zlotnik2013optimal, Kato2021optimization}, and
optimizing coupling functions~\cite{namura2024optimal}.
Control of mutual synchronization~\cite{namura2024optimal} is also achieved by feedback control,
where the control input is determined from the phase differences between the oscillators.

\begin{figure}[t]
\centering
\includegraphics[scale=0.35]{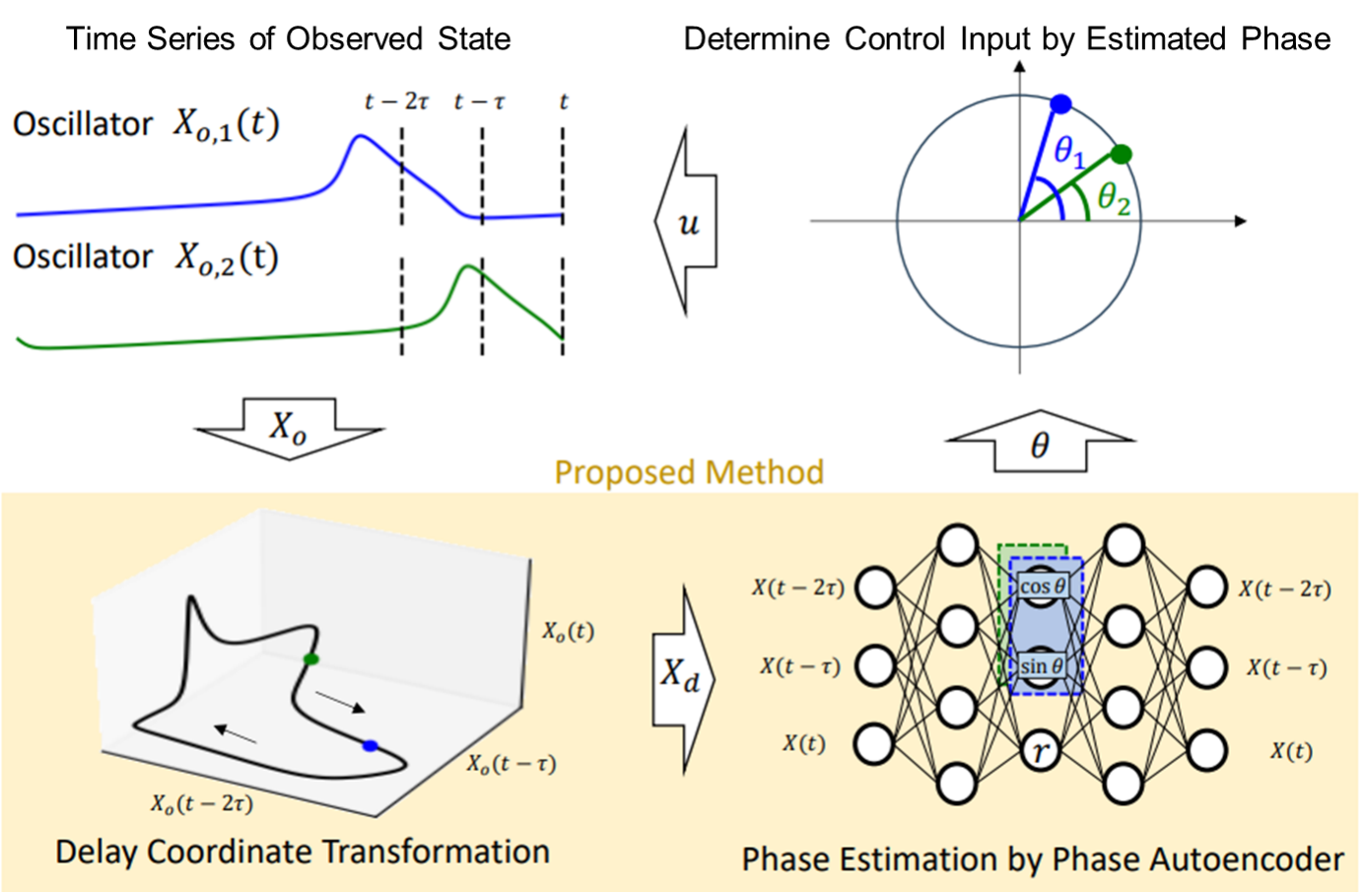}
\caption{Outline of this study.
The observed variables of the oscillator under partial observation are transformed into delay coordinates.
Then, the phase is estimated using a phase autoencoder for delay coordinates.
As an application, we synchronize two oscillators by feedback control using the estimated phase values.
}
\label{fig:CONCEPT} 
\end{figure}

Many control methods based on phase reduction,
such as those for mutual synchronization~\cite{namura2024optimal},
assume that the phase of the oscillator is available.
In real-world problems,
detailed mathematical models are often unavailable, 
and the asymptotic phase needs to be estimated in real time from time-series data as discussed in~\cite{netoff2012experimentally,Namura2022,kralemann2008phase,yawata2024phase,yamamoto2025gaussian,yawata2025phase}.
Although many studies assume that complete state variables can be observed,
it is often the case that not all the state variables can be observed and the asymptotic phase cannot be determined as a function of the complete state variables.
Therefore, it is necessary to develop a method for estimating the asymptotic phase of the oscillator
from time series under partial observation.
One typical method is to use 
the Hilbert transform to estimate the asymptotic phase from time series~\cite{matsuki2023extended}.
%
However, this method cannot be used for real-time control because it uses the entire time-series data to calculate the Hilbert transform.

In this study, we propose a machine-learning-based method for estimating the phase of a limit-cycle oscillator in real time under partial observation by extending the phase autoencoder proposed in~\cite{yawata2024phase}, 
which is an autoencoder-based method to characterize the oscillator state in the latent space by the asymptotic phase
(Figure~\ref{fig:CONCEPT}).
\INS{Compared to our previous study~\cite{yawata2024phase}, in which we could observe all the system variables,}
our present method introduces delay coordinates to the phase autoencoder to achieve phase estimation from time series data under partial observation.

This paper is organized as follows. 
In Sec.~\ref{sec:phase_reduction}, we briefly explain 
phase reduction theory.
In Sec.~\ref{sec:phase_autoencoder}, we review the phase autoencoder and present a method for estimating the oscillator phase using delay coordinates.
In Sec.~\ref{sec:control}, we evaluate the effect of a weak input on the estimation accuracy of the phase and present a feedback control method for mutual synchronization of two oscillators.
In Sec.~\ref{sec:experiment}, we verify the proposed method using numerical examples, and Sec.~\ref{sec:conclusions} concludes the paper.


\section{Phase reduction}\label{sec:phase_reduction}

We consider a limit-cycle oscillator with a $d_X$-dimensional state vector ${\bm X}(t) \in \mathbb{R}^{d_X}$ described by the following ordinary differential equation (ODE):
\begin{equation}\label{eq:main_ode}
    \dot{\bm X}(t) = {\bm F}({\bm X}(t)),
\end{equation}
where $t$ is the time and ${\bm F} : {\mathbb R}^{d_X} \to  {\mathbb R}^{d_X}$ is a smooth vector field that represents the oscillator dynamics. 
We assume that Eq.~\eqref{eq:main_ode} has an exponentially stable limit-cycle solution $\tilde{\bm{\mathcal{X}}}(t)$ with period $T$ and natural frequency $\omega = 2\pi/T$, satisfying $ \tilde{\bm{\mathcal{X}}}(t + T) = \tilde{\bm{\mathcal{X}}}(t) $.

To characterize the oscillator state ${\bm X}$, we introduce the asymptotic phase function $\Theta: A \to [0, 2\pi)$, where $A \subseteq {\mathbb R}^{d_X}$ is the basin of the limit cycle.
First, the phase of a state $\tilde{\bm{\mathcal{X}}}(t)$ on the limit cycle is defined as
\begin{equation}\label{eq:ph1}
    \Theta(\tilde{\bm{\mathcal{X}}}(t)) = \omega t \mod 2\pi,
\end{equation}
where $\tilde{\bm{\mathcal{X}}}(0)$ is chosen as the origin of the phase.
By this definition, a phase value between $0$ and $2\pi$ is assigned to the state $\tilde{\bm{\mathcal{X}}}(t)$ on the limit cycle, which increases constantly over $t$, i.e.,
\begin{equation}\label{eq:phase_fuction}
    \frac{d}{dt}\Theta(\tilde{\bm{\mathcal{X}}}(t)) = \omega,
\end{equation}
where the phase values $0$ and $2\pi$ are considered identical.
In what follows, we denote the state 
with the phase $\theta$ on the limit cycle as
\INS{${\bm {\mathcal{X}}}(\theta)=\tilde{\bm{\mathcal{X}}}(\theta/\omega)$}
, which is a $2\pi$-periodic function satisfying ${\bm {\mathcal{X}}}(\theta + 2\pi) = {\bm {\mathcal{X}}}(\theta)$.
Next, the definition of the phase is extended to the basin $A$, where the phase of the state ${\bm X} \in A$ is defined as $\theta$
if the state $\bm{X}$ converges to ${\bm {\mathcal{X}}}(\theta)$ over time.
By the above definition, the phase function $\Theta$ satisfies
\begin{equation}\label{eq:phase_fuction2}
    \frac{d}{dt}\Theta({\bm X}(t))
    = \text{grad}_{\bm X} \Theta({\bm X}) \cdot {\bm F}({\bm X})
    = \omega
\end{equation}
for any state ${\bm X}$ in $A$ obeying Eq.~(\ref{eq:main_ode}),
where $\text{grad}_{\bm X}$ represents the gradient at $\bm X$ and ``$\cdot$'' represents the scalar product of two vectors. 
We can define the phase value of the oscillator state $\bm{X} \in A$ as $\theta = \Theta(\bm{X})$.

The gradient of the phase function evaluated on the limit cycle is called the phase sensitivity function (PSF) $\bm Z: [0, 2\pi) \to {\mathbb R}^{d_X}$ and is defined as
\begin{gather}\label{eq:psf}
    {\bm Z}(\theta)=\text{grad}_{{\bm X}={\bm {\mathcal{X}}}(\theta)}\Theta({\bm X}).
\end{gather}
The PSF characterizes the linear response of the oscillator phase to a small perturbation given to the oscillator state with phase $\theta$ on the limit cycle.
The PSF can be calculated as a $2\pi$-periodic solution of the adjoint equation~\cite{nakao2016phase,Ermentrout2010}:
\begin{align}
\omega \frac{d}{d\theta} \bm{Z}(\theta) = -\bm{J}(\bm{\mathcal{X}}(\theta))^{\top} \bm{Z}(\theta),
\end{align}
where $\bm{J}(\bm{\mathcal{X}}(\theta))$ is the Jacobian matrix of the vector field $\bm{F}$ at the state 
$\bm{\mathcal{X}}(\theta)$ on the limit cycle and $\top$ represents matrix transpose.

\INS{From Eq.~(\ref{eq:phase_fuction2}),} the following normalization condition for the PSF should also be satisfied:
\begin{align}
\bm{Z}(\theta)\cdot \bm{F}(\bm{\mathcal{X}}(\theta))=\omega.
\end{align}

When the oscillator receives a sufficiently weak control input $\bm{u}(t) \in \mathbb{R}^{d_X}$ \INS{, which can be regarded as a perturbation, and it is described} as
\begin{equation}
    \dot{{\bm X}}(t) = {\bm F}({\bm X}(t)) + \bm{u}(t),
\end{equation}
the oscillator dynamics can be approximately reduced to a phase equation,
%
\begin{align}
    \dot{\theta}(t) = \omega + \bm{Z}(\theta) \cdot \bm{u}(t),
\end{align}
by phase reduction.


\section{Phase autoencoder for delay coordinates}\label{sec:phase_autoencoder}
\subsection{Phase autoencoder}

For unknown limit cycles, the phase of the oscillator should be estimated from time-series data. 
Phase autoencoder~\cite{yawata2024phase} is a method for estimating phase dynamics of limit-cycle oscillators, composed of an encoder $f_{enc}$ and a decoder $f_{dec}$, similar to regular autoencoders.
Phase autoencoder is trained so that the encoder embeds the system state in the latent space whose latent variables can be directly interpreted as the asymptotic phase and the deviation from the limit cycle, and the decoder reconstructs the original state from these latent variables.

The phase autoencoder has two principal loss functions.
One principal loss function is the reconstruction loss:
\begin{gather}
    L_{recon} = {\mathbb E} \left[ \left\|{\bm X}_t - f_{dec}(f_{enc}({\bm X}_t)) \right\|^2 \right],
    \label{eq:reconloss}
\end{gather}
where $\| \cdot \|$ represents the Euclidean norm, ${\bm X}_t$ represents the input data, and $\mathbb{E}$ represents the expectation with respect to the input data. 
In addition to the reconstruction loss, 
another principal loss function, the dynamical loss, is used to constrain the time evolution in the latent space.
We express the latent variable as a three-dimensional vector ${\bm Y}_t=[Y_{1,t}\;Y_{2,t}\;Y_{3,t}]^\top=f_{enc}({\bm X}_t)$.
%
Here, $Y_{1,t}$ and $Y_{2,t}$ are normalized so that $Y_{1,t}^2+Y_{2,t}^2=1$ in the encoding process.
The dynamical loss is:
\begin{gather}
    \footnotesize
     L_{dyn} = {\mathbb E} \left[ \left\|{\bm Y}_{t+\Delta t} -
     {\setlength{\arraycolsep}{1pt}\begin{bmatrix}
    \cos \omega \Delta t & -\sin \omega \Delta t & 0\\
     \sin \omega \Delta t & \cos \omega \Delta t & 0\\
    0 & 0 & e^{-\lambda\Delta t}
    \end{bmatrix}}{\bm Y}_t \right\|^2 \right],
\end{gather}
where $\omega$ and $\lambda$ are also the parameters to be learned,
corresponding to the natural frequency and the decay rate of the oscillator, respectively.
The time step $\Delta t$ of the input data is fixed in advance. 
This loss function requires that $Y_{1,t}$ and $Y_{2,t}$ evolve in the $Y_{1}$-$Y_{2}$ plane at a constant angular velocity $\omega$, and $Y_{3,t}$ approaches 0 asymptotically over time,
\INS{where $Y_{3,t}$ can be regarded as corresponding to the amplitude of deviation from the limit cycle.}
In this way, the unit circle 
in the $Y_{1}$-$Y_{2}$ plane corresponds to the limit cycle, 
and the two latent variables $Y_{1,t}, Y_{2,t}$ represent the asymptotic phase
as $\theta_t=\text{arctan}(Y_{1,t}, Y_{2,t})$.
Thus, the encoder $f_{enc}$ gives an approximation of the asymptotic phase function.
\INS{In addition to these two loss functions, an auxiliary loss function is used during training to stabilize it.
The training parameters are given in~\cite{yawata2024phase}.}

\subsection{Delay coordinate phase function}

In real-world systems, it is often the case that only a part of the oscillator state can be observed, 
which prevents the straightforward application of the phase autoencoder that requires the input data 
of the complete state variables.
In this study, we introduce a phase function in delay coordinates, which enables to estimate the asymptotic phase of the oscillator under partial observation. 
The phase autoencoder 
framework is thus generalized for the case with delay coordinates, where the complete state variables in the \INS{delay coordinates} are used as the input.

Let ${\bm X_o}\in {\mathbb R}^{d_{X_o}}$ be the observed variables and ${\bm X_h}\in {\mathbb R}^{d_{X_h}}$ be the hidden state.
The state variables are sorted so that the system state can be represented as
\begin{equation}
    {\bm X}(t) = \begin{bmatrix}
    {\bm X_o}(t) \\
    {\bm X_h}(t)
    \end{bmatrix}
    ,
\end{equation}
where $d_{X_o}$ and $d_{X_h}$ are the dimensions of the observed space and the hidden space, respectively, satisfying $d_{X_o} + d_{X_h} = d_X$.
Similarly, we assume that the PSF $\bm{Z}$ can be represented by \INS{$\bm{Z}_o\in {\mathbb R}^{d_{X_o}}$ and $\bm{Z}_h\in {\mathbb R}^{d_{X_h}}$} corresponding to the observed and hidden state, respectively, as follows:
\begin{equation}
    {\bm Z}(\theta) = \begin{bmatrix}
    {\bm Z_o}(\theta) \\
    {\bm Z_h}(\theta)
    \end{bmatrix}
    .
\end{equation}
In this study, we assume that ${\bm Z_o}(\theta)$ is not constantly zero, that is, the oscillator responds to the input given to ${\bm X}_o$.

Although it is impossible to obtain the asymptotic phase only from the current observed state ${\bm{X}_o}$, it is possible to obtain the asymptotic phase by using delay embedding.
We call a limit cycle in the delay coordinate system a DC limit cycle $\bm{\mathcal{X}}_{dc}(\theta)$,
which is expressed by 
\begin{align}
    \bm{{\mathcal{X}}}_{dc}(\theta) = \left[
        {{\mathcal{X}}}_{o}(\theta) \;
        {{\mathcal{X}}}_{o}\left( \theta-\frac{\omega}{K}\tau \right) \;
        \cdots \;
        {{\mathcal{X}}}_{o}(\theta-\omega\tau))
    \right]^\top
    ,
\end{align}
where $\mathcal{X}_o(\theta) \in \mathbb{R}$ is the observed variable 
of the state $\bm{\mathcal{X}}(\theta)$ on the limit cycle.
Here, $\tau$ is the maximum time delay and $K$ is the number of time steps for the delay used for the delay coordinate.
Next, we assume that the asymptotic phase of the oscillator in the delay coordinate system is defined as
\begin{align}
    \Theta({\bm X}(t)) &= \Theta_H \left( {\bm X}_{o}(t),{\bm X}_{o}\left(t-\frac{1}{K}\tau\right),\cdots, {\bm X}_{o}(t-\tau) \right).
\end{align}
We call 
$\Theta_H$ the delay coordinate phase function (DCPF).
Furthermore, we assume that the gradient function ${\bm g}(\theta) = [g_0\;g_1\cdots g_K]^\top = \text{grad}_{\bm{{\mathcal{X}}}_{dc}(\theta)} \Theta_H$ can be calculated and each component is $O(1)$ in the vicinity of the DC limit cycle\INS{, which is typical for the oscillators observed in practical setups}.

The proposed phase autoencoder takes the input data from the reconstructed system states in delay coordinates and is trained in the same way as the original phase autoencoder.


\section{Feedback control of the oscillator} \label{sec:control}
\subsection{Phase estimation under weak control inputs}

In feedback control of unknown oscillators, the phase must be estimated to determine the feedback, while the control input is being applied to the oscillator.
Therefore, we need to examine the effect of the control input on the accuracy of phase estimation, because the learning of DCPF assumes no control input.

In what follows, we consider the case that only one variable $X_{o}^{u}$ can be observed, i.e., $d_{X_o}=1$, and the control input $u_o(t) \in \mathbb{R}$ is applied only to this variable:
\begin{align}
\begin{aligned}
    \dot{\bm X}(t) &= {\bm F}({\bm X}(t)) +{\bm u(t)},\\
    {\bm u(t)}&=\begin{bmatrix}
    {u_o}(t) \\
    {\bm 0}
    \end{bmatrix}
    .
\label{eq:main_ode_adu}
\end{aligned}
\end{align}

Since DCPF is learned under the assumption without the input control,
we need to examine the estimation error of DCPF under the control input 
when the oscillator state is near the limit cycle.
If 
${\mathcal{X}}_{o}(\theta)=X^u_o(t)$, {then} ${\mathcal{X}}_{o}(\theta-\frac{k\omega}{K}\tau)=X^u_o(t-\frac{k}{K}\tau)$ $(k=1,\cdots K)$ is also expected, but this is not satisfied under the influence of the control input.
The resulting shift $\Delta X_k={\mathcal{X}}_{o}(\theta-\frac{k\omega}{K}\tau)-X_{o}^{u}(t-\frac{k}{K}\tau)$ causes the phase estimation error $\Delta \theta$.

\begin{thm}
    When the oscillator state is sufficiently close to the limit cycle and the intensity $\epsilon$ of the control input $u_o$ is sufficiently small,
    the DCPF estimation error $\Delta \theta$ is of $O(\epsilon K)$
    with respect to the dimension of the delay coordinate system $K$.
\end{thm}

\begin{proof}
When the intensity of the control input is $\epsilon$,
the shift $\Delta X_k$ is of $O(\epsilon)$.
The estimation error $\Delta \theta$ is given by
\begin{align}
\begin{aligned}
    \Delta \theta &=  \Theta_H \left( X_{o}(t),X_{o}\left(t-\frac{1}{K}\tau\right),\cdots, X_{o}(t-\tau) \right)\\
    &-\Theta_H \left(
        {{\mathcal{X}}}_{o}(\theta) \;
        {{\mathcal{X}}}_{o}\left( \theta-\frac{\omega}{K}\tau \right) \;
        \cdots \;
        {{\mathcal{X}}}_{o}(\theta-\omega\tau)\right).
\end{aligned}
\end{align}
In the close vicinity of the DC limit cycle,
\begin{align}
\begin{aligned}
|\Delta \theta| &= \left| \sum_{k=0}^{K} g_k \Delta X_k \right|+ O(\epsilon^2)\\
&\leq (K+1) g_{\text{Max}} \Delta X_{\text{Max}}+ O(\epsilon^2),
\end{aligned}
\end{align}
where $g_{\text{Max}}$ is the maximum absolute value within all the components of $\bm{g}(\theta)$ and 
$\Delta X_{\text{Max}}$ is the maximum value among $[\Delta X_0 \;\Delta X_1\cdots,\Delta X_K]$.
This proves Theorem 1, because we assume that $g_{\text{Max}}$ is of $O(1)$
and $\Delta X_{\text{Max}}$ is of $O(\epsilon)$, and the higher-order terms can be ignored.
\end{proof}

\subsection{Feedback control using DCPF for synchronization}

As an application, we consider a simple feedback control using DCPF for synchronization.
We consider two oscillators with identical properties whose states are denoted as $\bm{X}_{1}$ and ${\bm X}_{2}$, and assume that the control input $\bm u$ is applied only to the observed variable 
$X_{o,1}$ or $X_{o,2}$ of each oscillator as follows:
\begin{align}
    \dot{\bm X}_1 &= {\bm F}({\bm X}_1) +\begin{bmatrix}
    u_{o,1}(t) \\
    {\bm 0}
    \end{bmatrix},
    \\
    \dot{\bm X}_2 &= {\bm F}({\bm X}_2) +\begin{bmatrix}
    u_{o,2}(t) \\
    {\bm 0}
    \end{bmatrix}, \label{eq:sync}
\end{align}
where ${u}_{o,1}, {u}_{o,2} \in {\mathbb R}$ are the control inputs for synchronization.
We take these feedback control inputs as
\begin{align}
    u_{o,1} = \epsilon \sin (\theta_2 - \theta_1)Z_o(\theta_1),\cr
    u_{o,2} = \epsilon \sin (\theta_1 - \theta_2)Z_o(\theta_2).
    \label{eq:control}
\end{align}
By phase reduction, the dynamics of the phases $\theta_{1,2}$ are approximately given by 
\begin{align}
    \dot{\theta}_1(t) &= \omega + Z_o(\theta_1) u_{o,1}(t), \\
    \dot{\theta}_2(t) &= \omega + Z_o(\theta_2) u_{o,2}(t).
\end{align}

The control objective is to achieve in-phase synchronization of the two oscillators, i.e., to achieve $\theta_1(t) = \theta_2(t)$.
If we take the input as in Eq.~\eqref{eq:control}, 
the dynamics of the phase difference $\phi(t) = \theta_1(t) - \theta_2(t)$ can be represented as 
\begin{align}
\begin{aligned}
    \dot{\phi}(t) 
    &=u_{o,1}(t) Z_o(\theta_1(t)) - u_{o,2}(t) Z_o(\theta_2(t))\\
    &=\epsilon \sin (\theta_2(t) - \theta_1(t))Z_{o}^2(\theta_1(t)) \\
    &\qquad- \epsilon \sin (\theta_1(t) - \theta_2(t))Z_{o}^2(\theta_2(t))\\
    &=- \epsilon \sin \phi(t) (Z_{o}^2(\theta_1(t))+Z_{o}^2(\theta_2)(t)).
    \label{eq:phase_difference}
\end{aligned}
\end{align}
Since $Z_0^2(\theta_1)+Z_0^2(\theta_2) \geq 0$ and not constantly zero, $\phi(t)$ converges to $0$ and the two oscillators eventually exhibit in-phase synchronization.

In practice, we estimate the asymptotic phases of the oscillators by DCPF as 

\begin{gather}
    \hat{\theta_i} = \Theta_H \left(X_{i,o}^{u}(t),X_{i,o}^{u}\left(t-\frac{1}{K}\tau\right),\cdots X_{i,o}^{u}(t-\tau) \right)
\end{gather}
for $i = 1, 2$,
and determine the control input based on the estimated phases in a data-driven manner as follows:
\begin{gather}
    u_{o,1} = \epsilon \sin (\hat{\theta}_2 - \hat{\theta}_1)Z_o(\hat{\theta}_1),\\
    u_{o,2} = \epsilon \sin (\hat{\theta}_1 - \hat{\theta}_2)Z_o(\hat{\theta}_2).
\end{gather}
We first estimate the phase of the oscillator by DCPF from the time series data of the observed variables and then determine the feedback control input at each time step.


\section{Numerical Experiments}\label{sec:experiment}

In this section, we verify the validity of the present method for estimating DCPF by numerical simulations of two types of limit-cycle oscillators as examples.
We also show that the present method of feedback control under partial observation is useful for realizing mutual synchronization.
For both oscillators,
we perform
direct numerical simulations 
using the same procedures and parameters 
as in~\cite{yawata2024phase};
the time-series data are obtained by evolving the system over time $3T$ from randomly selected five hundred initial points.
\INS{In addition, $K$ was set so that the dimension of the delay-coordinate system matched that of the original system, and $\tau$ was set to about $T/20$ in each experiment.}

\subsection{Stuart-Landau oscillator}

First, we apply our method to the Stuart-Landau (SL) oscillator~\cite{Kuramoto2003},
which is a normal form of the supercritical Hopf bifurcation described by
\begin{gather}
\begin{bmatrix}
\dot{x}_1\\
\dot{x}_2\\
\end{bmatrix}=
\begin{bmatrix}
x_1 - \alpha x_2 - (x_1-\beta x_2)(x_1^2+x_2^2)\\
\alpha x_1 + x_2 - (\beta x_1+x_2)(x_1^2+x_2^2)\\
\end{bmatrix}.
\end{gather}
Here, $\bm{X} = [x_1\;x_2]^{\top}$ is the state vector and $(\alpha, \beta)$ are the system parameters.
The limit-cycle solution is $\tilde{\bm{\mathcal{X}}}_{0}(t) = [\cos \omega t\; \sin \omega t]^{\top}$, which is a unit circle on the $x_1$-$x_2$ plane.
The natural frequency is given by $\omega = \alpha - \beta$ and the oscillation period is $T = 2\pi/(\alpha - \beta)$.

We conducted 
numerical experiments under the assumption that only $x_1$ is observed.
We set the parameters as $\alpha=2\pi$ and $\beta=1$, which yields $T \simeq 1.19$.
We trained the autoencoder and then estimated the DCPF of the SL oscillator.
We set the number of time steps for the delay coordinate as $K=1$
and the maximum time delay as $\tau=0.4$.

First, we examined the accuracy of estimating the asymptotic phase in delay coordinates by the proposed autoencoder.
The true asymptotic phase function of the SL oscillator is analytically given by~\cite{nakao2016phase}
\begin{gather}
    \Theta(x_1, x_2) = \text{arctan}\left(x_2/x_1 \right) - \log \sqrt{x_1^2+x_2^2}
\end{gather}
if we choose $(x_1, x_2) = (1, 0)$ as the state of the origin of the phase.
Figure~\ref{fig:SL_PF} compares the estimated DCPF $\hat{\theta}$ with the true phase $\theta$.
We confirm that DCPF is reproduced well around the limit cycle.
For the oscillator states far from the limit cycle, the estimation accuracy  degraded due to insufficient data points.

\begin{figure}[t]
\includegraphics[scale=0.27]{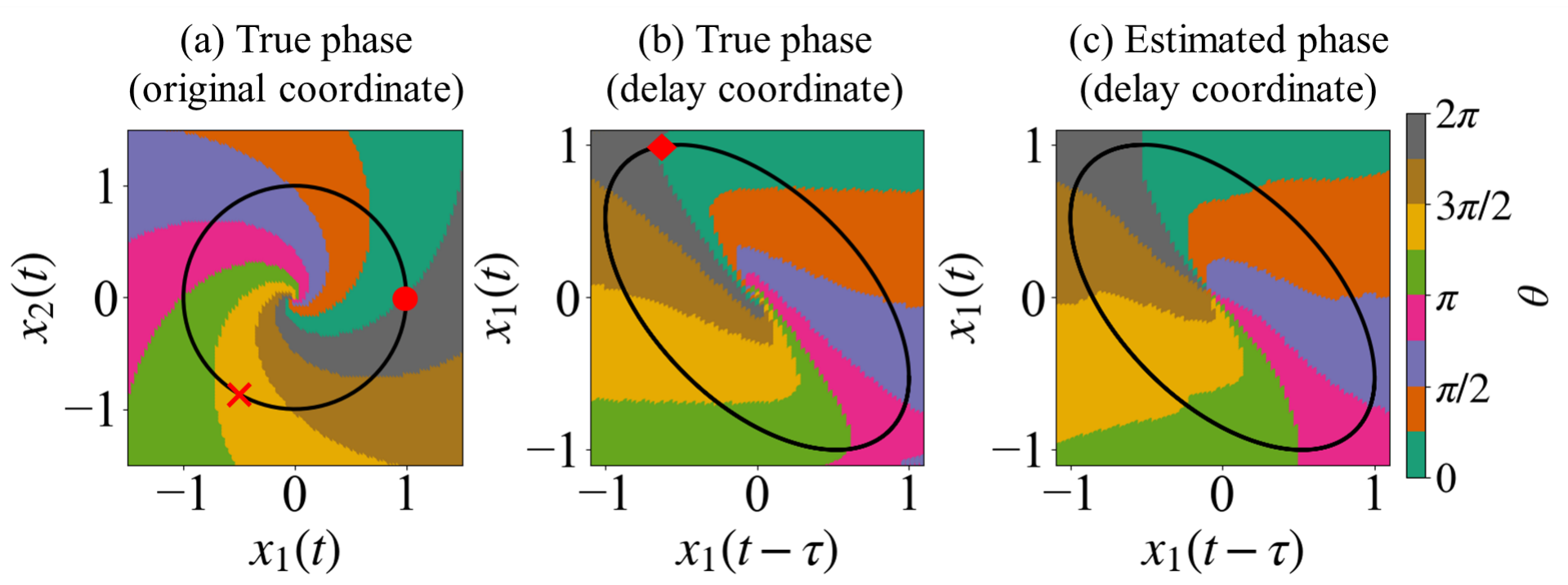}
\caption{Asymptotic phase of the SL oscillator in the original coordinates and delay coordinates.
(a) True phase function in the original coordinates. 
(b) True phase function in the delay coordinates. 
(c) Estimated DCPF by the proposed autoencoders.
The colors represent the phase from $0$ to $2\pi$ (discretized for visual clarity), where $(x_1,x_2)=(1,0)$ is chosen as the origin of the phase with $\theta=0$ (red circle).
The red cross symbol represents $\bm{\mathcal{X}}(-\omega\tau)$, and the red diamond symbol represents $\bm{\mathcal{X}}_{dc}(0)$.
The black line in each figure represents the true limit cycle.
}
\label{fig:SL_PF} 
\end{figure}

Next, we evaluated the phase estimation error caused by the control input to illustrate the theorem.
We used a periodic input represented as $u(t)=\epsilon \sin(\omega_ut)$.
We set $\omega_u$ as $1.05 \times \omega$, for which the oscillator was expected to synchronize with the input,
because the frequency mismatch between the oscillator and the periodic input is relatively small. 
Figure~\ref{fig:SL_p} shows the estimated phase for two different values of $\epsilon$. 
In Figs.~\ref{fig:SL_p}(a) and (b), the black dashed lines represent the true phase calculated analytically using both $x_1$ and $x_2$, 
while the red lines represent the estimated results.
Figure~\ref{fig:SL_p}(c) represents the time average of the absolute phase estimation error.
It can be seen that the error ${\mathbb E}[|\Delta \theta|]$ becomes smaller and almost converges to zero as $\epsilon$ is decreased.

\begin{figure}[t]
\includegraphics[scale=0.48]{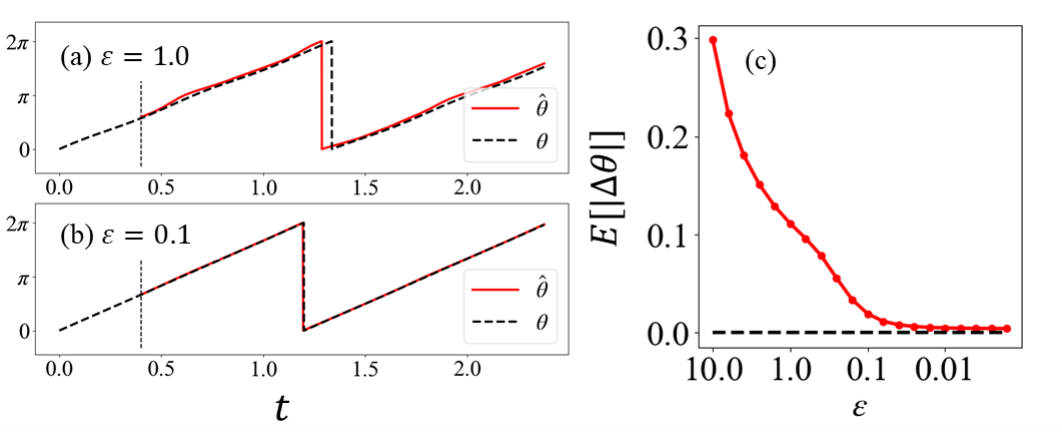}
\caption{
Impact of control inputs on phase estimation over time.
(a) and (b) show the estimated phase (red solid curve) and the true phase (black dashed curve) for $\epsilon=1.0$ and $0.1$, respectively.
Figure (c) shows the phase estimation error for different values of $\epsilon$.
}
\label{fig:SL_p} 
\end{figure}

\subsection{Hodgkin-Huxley model}

As the second example, we applied the proposed method to the Hodgkin-Huxley (HH) model of spiking neurons~\cite{hodgkin1952quantitative},
given by
\begin{gather}
\frac{d}{dt}\begin{bmatrix}
V\\
m\\
h\\
n\\
\end{bmatrix}=
\begin{bmatrix}
(G_{Na}m^3h(E_{Na}-V) + G_Kn^4(E_K-V) \\+ G_L(E_L-V) + I)/C\\
\alpha_m(V)(1.0-m)-\beta_m(V)m\\
\alpha_h(V)(1.0-h)-\beta_h(V)h\\
\alpha_n(V)(1.0-n)-\beta_n(V)n\\
\end{bmatrix},
\end{gather}
where $V$ is the membrane potential, $(m, h, n)$ are the channel variables,
and the voltage-dependent rate constants are given by
\begin{align}
\begin{aligned}
\alpha_m(V) &= 0.1 (V+40)(1-\text{exp}(-(V+40)/10)), \\
\beta_m(V) &= 0.1 \text{exp}(-(V+65)/18), \\
\alpha_h(V) &= 0.07 \text{exp}(-(V+65)/20), \\
\beta_h(V) &= 1/(1+\text{exp}(-V+35)/10)), \\
\alpha_n(V) &= 0.01 (V+55)/(1-\text{exp}(-(V+55)/10)), \\
\beta_n(V) &= 0.125 \text{exp}(-(V+65)/80).
\end{aligned}
\end{align}
We set the parameters as $C=1.0, G_{N_a}=120.0, G_K=36.0, G_L=0.3, E_{N_a}=50.0, E_K=-77.0$, and $E_L=-54.4$,
at which the HH model possesses a stable limit cycle with period $T \simeq 10.12$.
\INS{Here, $V$ is the observable physical quantity that can be observed in a practical setting.}
We conducted the numerical experiment assuming that only $V$ is observed.
We set the time steps for the delay as $K=3$
and the maximum time delay as $\tau=1.5$.

First, we examined whether the correct phase values were assigned to the oscillator states on the limit cycle.
We sampled the oscillator state $\bm{\mathcal{X}}_{dc}(\theta)$ 
on the limit cycle
represented in delay coordinates. 
We then estimated the phase values of the sampled 
states by the autoencoder to confirm if 
DCPF is properly embedded.
Figure~\ref{fig:HH_LC}(a) compares the 
true phase $\theta$ and the estimated phase $\hat{\theta}$ from $\bm{\mathcal{X}}_{dc}(\theta)$, showing good agreement.
We also verified that the limit cycle is appropriately embedded in the delay coordinates. 
Figure~\ref{fig:HH_LC}(b) compares the original DC limit cycle and the decoder's output on the $V(t)$-$V(t-\tau)$ plane, which is the projection of the reconstructed state in the delay coordinate, showing good agreement.

\begin{figure}[t]
\centering
\includegraphics[scale=0.26]{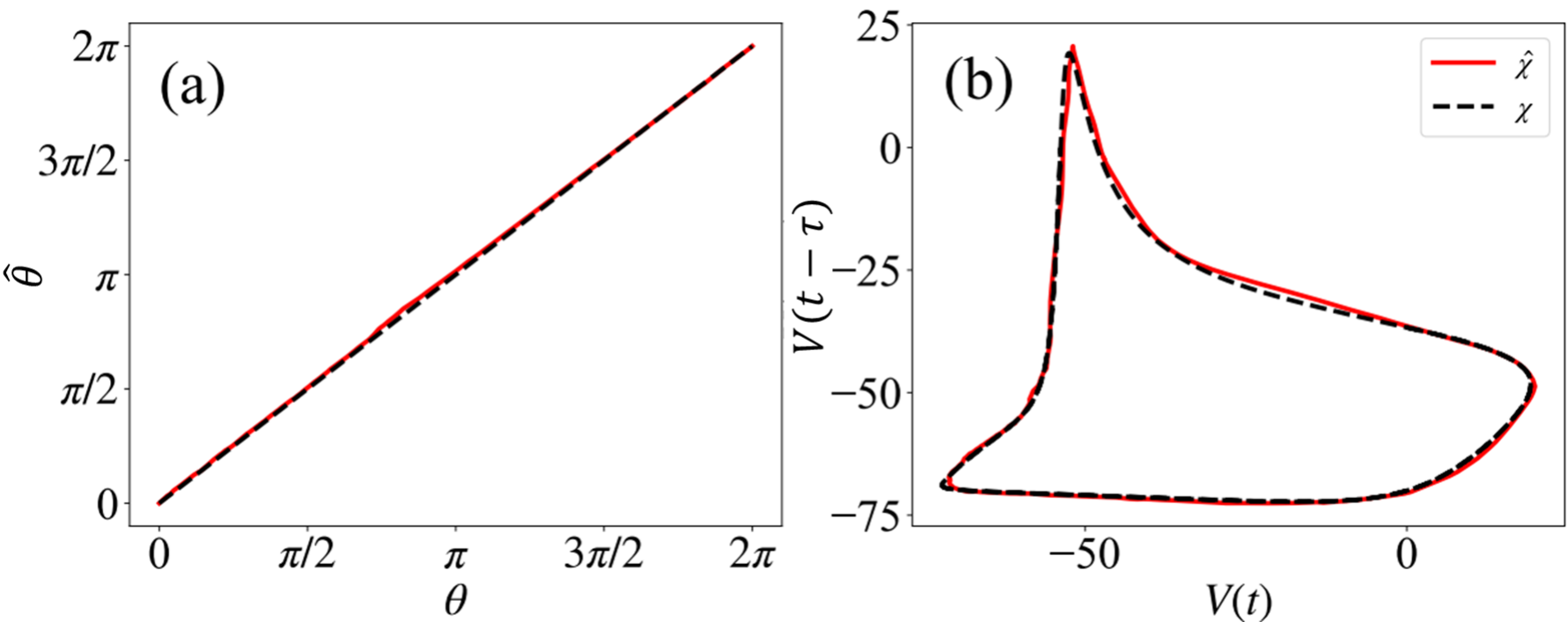}
\caption{
(a) Comparison of the estimated phase (red solid line) $\hat{\theta}$ and the true phase (black dotted line) $\theta$ of the HH oscillator.
(b) Reconstructed DC limit cycle by the autoencoder $\hat{\bm{\mathcal{X}}}_{dc}$ (red solid curve) compared with the true DC limit cycle $\bm{\mathcal{X}}_{dc}$ (black dotted curve) in $V(t)$-$V(t-\tau)$ coordinates. 
}
\label{fig:HH_LC} 
\end{figure}

Next, we tested the feedback control using the estimated DCPF. 
We estimated the phase only from $V$
and determined the control input 
using the estimated phase and PSF to 
synchronize the two HH oscillators.
The true PSF was calculated by the adjoint equation.
Since, in practice, the adjoint equation is unavailable for unknown limit-cycle oscillators under partial observation, the PSF should be obtained using the direct method, i.e., by perturbing the target oscillator by impulses.
In this numerical experiment, the PSF was calculated using complete states ${\bm X}(t)$ and the adjoint method.

Figure~\ref{fig:HH_SYNC}(a) and (b) show the observed $V_{1,2}$ 
and the estimated phases, respectively.
Figure~\ref{fig:HH_SYNC}(c) shows the control inputs, and Fig. \ref{fig:HH_SYNC}(d) shows the estimated phase difference and the true phase difference from the initial state ($\phi(0)\simeq 4/5\pi$).
Unlike the SL model, 
the true phase value of the oscillator state cannot be obtained analytically for
the HH model.
Here, we used the phase estimated by the phase autoencoder using all state variables as the true phase.
We can confirm that, even under the feedback control input, the phase difference was accurately estimated and data-driven in-phase synchronization was achieved.

\begin{figure}[t]
\includegraphics[scale=0.35]{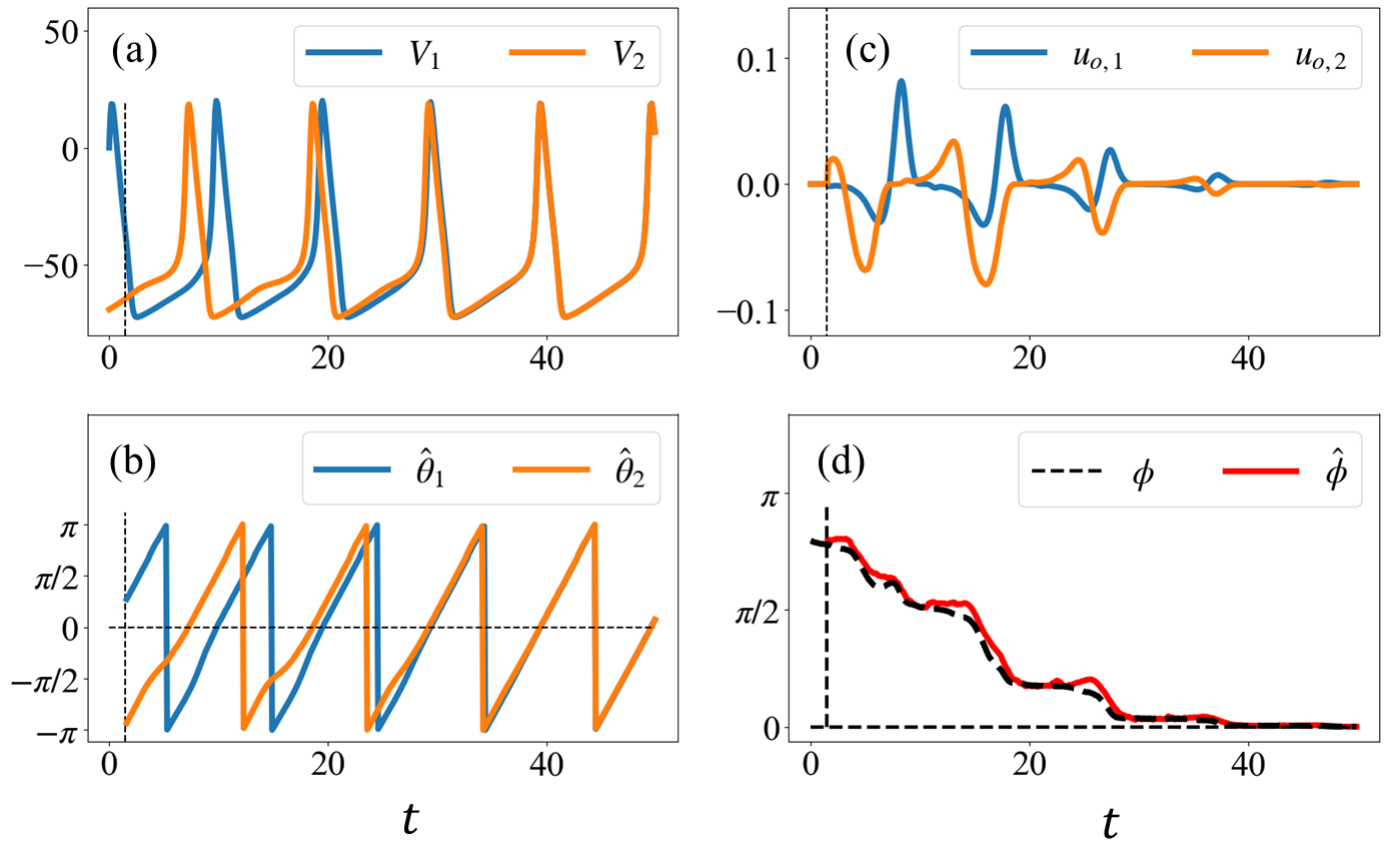}
\caption{
Feedback control for synchronization of the two HH oscillators.
(a) $V_{1}$ (blue) and $V_{2}$ (orange).
(b) Estimated phase $\hat{\theta}_1$ (blue) and $\hat{\theta}_2$ (orange).
(c) Control input $u_{o,1}$ and $u_{o,2}$.
(d) Comparison of the estimated phase difference (red solid line) and the true phase difference (black dotted line). 
}
\label{fig:HH_SYNC} 
\end{figure}


\section{Conclusions}\label{sec:conclusions}

We proposed a method for estimating the asymptotic phase of limit-cycle oscillators under partial observation in real time using delay coordinates.
We evaluated the upper bound of the phase estimation error near the limit cycle and applied the proposed method to a simple feedback control for mutual synchronization of two limit-cycle oscillators.
We verified that the proposed method correctly estimates the asymptotic phase with and without the control input, and also that the feedback control works well for synchronization.
\INS{A more detailed analysis of the impact of the phase estimation error introduced by our methods is a future work.}


\bibliographystyle{ieeetran}
\bibliography{dcph}

\providecommand{\noopsort}[1]{}\providecommand{\singleletter}[1]{#1}%
\begin{thebibliography}{10}
\providecommand{\url}[1]{#1}
\csname url@samestyle\endcsname
\providecommand{\newblock}{\relax}
\providecommand{\bibinfo}[2]{#2}
\providecommand{\BIBentrySTDinterwordspacing}{\spaceskip=0pt\relax}
\providecommand{\BIBentryALTinterwordstretchfactor}{4}
\providecommand{\BIBentryALTinterwordspacing}{\spaceskip=\fontdimen2\font plus
\BIBentryALTinterwordstretchfactor\fontdimen3\font minus \fontdimen4\font\relax}
\providecommand{\BIBforeignlanguage}[2]{{%
\expandafter\ifx\csname l@#1\endcsname\relax
\typeout{** WARNING: IEEEtran.bst: No hyphenation pattern has been}%
\typeout{** loaded for the language `#1'. Using the pattern for}%
\typeout{** the default language instead.}%
\else
\language=\csname l@#1\endcsname
\fi
#2}}
\providecommand{\BIBdecl}{\relax}
\BIBdecl

\bibitem{Winfree1980}
A.~T. Winfree, \emph{The geometry of biological time}.\hskip 1em plus 0.5em minus 0.4em\relax Springer, 1980.

\bibitem{Pikovsky2001}
A.~Pikovsky, M.~Rosenblum, and J.~Kurths, \emph{Synchronization: {A} universal concept in nonlinear sciences}.\hskip 1em plus 0.5em minus 0.4em\relax Cambridge University Press, 2001.

\bibitem{Kuramoto2003}
Y.~Kuramoto, \emph{Chemical oscillations, waves, and turbulence.}\hskip 1em plus 0.5em minus 0.4em\relax Springer, 2003.

\bibitem{Kralemann2013}
B.~Kralemann, M.~Fr\"uhwirth, A.~Pikovsky, M.~Rosenblum, T.~Kenner, J.~Schaefer, and M.~Moser, ``In vivo cardiac phase response curve elucidates human respiratory heart rate variability,'' \emph{Nature Communications}, vol.~4, 2013.

\bibitem{Stankovski2015}
T.~Stankovski, V.~Ticcinelli, P.~V. McClintock, and A.~Stefanovska, ``Coupling functions in networks of oscillators,'' \emph{New Journal of Physics}, vol.~17, 2015.

\bibitem{nakao2016phase}
H.~Nakao, ``Phase reduction approach to synchronisation of nonlinear oscillators,'' \emph{Contemporary Physics}, vol.~57, no.~2, pp. 188--214, 2016.

\bibitem{Ermentrout2010}
G.~B. Ermentrout and D.~H. Terman, \emph{Mathematical foundations of neuroscience}, 2010, vol.~35.

\bibitem{monga2019phase}
B.~Monga, D.~Wilson, T.~Matchen, and J.~Moehlis, ``Phase reduction and phase-based optimal control for biological systems: a tutorial,'' \emph{Biological Cybernetics}, vol. 113, no. 1-2, pp. 11--46, 2019.

\bibitem{ermentrout2019recent}
B.~Ermentrout, Y.~Park, and D.~Wilson, ``Recent advances in coupled oscillator theory,'' \emph{Philosophical Transactions of the Royal Society A}, vol. 377, no. 2160, p. 20190092, 2019.

\bibitem{Dorfler2012synchronization}
F.~D\"{o}rfler and F.~Bullo, ``Synchronization and transient stability in power networks and nonuniform {K}uramoto oscillators,'' \emph{SIAM Journal on Control and Optimization}, vol.~50, no.~3, pp. 1616--1642, 2012.

\bibitem{Namura2025central}
N.~Namura and H.~Nakao, ``A central pattern generator network for simple control of gait transitions in hexapod robots based on phase reduction,'' \emph{Nonlinear Dynamics}, vol. 113, p. 10105^^e2^^80^^9310125, Jan 2025.

\bibitem{wilson2017spatiotemporal}
D.~Wilson and J.~Moehlis, ``Spatiotemporal control to eliminate cardiac alternans using isostable reduction,'' \emph{Physica D: Nonlinear Phenomena}, vol. 342, pp. 32--44, 2017.

\bibitem{glass1988clocks}
L.~Glass and M.~C. Mackey, \emph{From clocks to chaos: The rhythms of life}.\hskip 1em plus 0.5em minus 0.4em\relax Princeton University Press, 1988.

\bibitem{shirasaka2020phase}
S.~Shirasaka, W.~Kurebayashi, and H.~Nakao, ``Phase-amplitude reduction of limit cycling systems,'' \emph{The {K}oopman Operator in Systems and Control: Concepts, Methodologies, and Applications}, pp. 383--417, 2020.

\bibitem{wilson2022adaptive}
D.~Wilson, ``An adaptive phase-amplitude reduction framework without o($\epsilon$) constraints on inputs,'' \emph{SIAM Journal on Applied Dynamical Systems}, vol.~21, no.~1, pp. 204--230, 2022.

\bibitem{mauroy2013isostables}
A.~Mauroy, I.~Mezi{\'c}, and J.~Moehlis, ``Isostables, isochrons, and {K}oopman spectrum for the action--angle representation of stable fixed point dynamics,'' \emph{Physica D: Nonlinear Phenomena}, vol. 261, pp. 19--30, 2013.

\bibitem{wilson2016isostable}
D.~Wilson and J.~Moehlis, ``Isostable reduction of periodic orbits,'' \emph{Physical Review E}, vol.~94, no.~5, p. 052213, 2016.

\bibitem{moehlis2006optimal}
J.~Moehlis, E.~Shea-Brown, and H.~Rabitz, ``Optimal inputs for phase models of spiking neurons,'' \emph{Journal of Computational and Nonlinear Dynamics}, vol.~1, no.~4, pp. 358--367, 2006.

\bibitem{Monga2019optimal}
B.~Monga and J.~Moehlis, ``Optimal phase control of biological oscillators using augmented phase reduction,'' \emph{Biological Cybernetics}, vol. 113, no.~1, pp. 161--178, 2019.

\bibitem{namura2024optimal2}
N.~Namura and H.~Nakao, ``Optimal phase control of limit-cycle oscillators with strong inputs through phase-amplitude reduction,'' in \emph{2024 IEEE 63rd Conference on Decision and Control (CDC)}, 2024, pp. 4003--4009.

\bibitem{Zlotnik2013optimal}
A.~Zlotnik, Y.~Chen, I.~Z. Kiss, H.-A. Tanaka, and J.-S. Li, ``Optimal waveform for fast entrainment of weakly forced nonlinear oscillators,'' \emph{Physical Review Letters}, vol. 111, p. 024102, 2013.

\bibitem{Kato2021optimization}
Y.~Kato, A.~Zlotnik, J.-S. Li, and H.~Nakao, ``Optimization of periodic input waveforms for global entrainment of weakly forced limit-cycle oscillators,'' \emph{Nonlinear Dynamics}, vol. 105, no.~3, pp. 2247--2263, 2021.

\bibitem{namura2024optimal}
N.~Namura and H.~Nakao, ``Optimal coupling functions for fast and global synchronization of weakly coupled limit-cycle oscillators,'' \emph{Chaos, Solitons \& Fractals}, vol. 185, p. 115080, 2024.

\bibitem{netoff2012experimentally}
T.~Netoff, M.~A. Schwemmer, and T.~J. Lewis, ``Experimentally estimating phase response curves of neurons: theoretical and practical issues,'' \emph{Phase Response Curves in Neuroscience: Theory, Experiment, and Analysis}, pp. 95--129, 2012.

\bibitem{Namura2022}
N.~Namura, S.~Takata, K.~Yamaguchi, R.~Kobayashi, and H.~Nakao, ``Estimating asymptotic phase and amplitude functions of limit-cycle oscillators from time series data,'' \emph{Physical Review E}, vol. 106, p. 14204, 2022.

\bibitem{kralemann2008phase}
B.~Kralemann, L.~Cimponeriu, M.~Rosenblum, A.~Pikovsky, and R.~Mrowka, ``Phase dynamics of coupled oscillators reconstructed from data,'' \emph{Physical Review E}, vol.~77, no.~6, p. 066205, 2008.

\bibitem{yawata2024phase}
K.~Yawata, K.~Fukami, K.~Taira, and H.~Nakao, ``Phase autoencoder for limit-cycle oscillators,'' \emph{Chaos: An Interdisciplinary Journal of Nonlinear Science}, vol.~34, no.~6, 2024.

\bibitem{yamamoto2025gaussian}
T.~Yamamoto, H.~Nakao, and R.~Kobayashi, ``Gaussian process phase interpolation for estimating the asymptotic phase of a limit cycle oscillator from time series data,'' \emph{Chaos, Solitons \& Fractals}, vol. 191, p. 115913, 2025.

\bibitem{yawata2025phase}
K.~Yawata, R.~Sakuma, K.~Fukami, K.~Taira, and H.~Nakao, ``Phase autoencoder for rapid data-driven synchronization of rhythmic spatiotemporal patterns,'' \emph{arXiv preprint arXiv:2506.12777}, 2025.

\bibitem{matsuki2023extended}
A.~Matsuki, H.~Kori, and R.~Kobayashi, ``An extended {H}ilbert transform method for reconstructing the phase from an oscillatory signal,'' \emph{Scientific Reports}, vol.~13, no.~1, p. 3535, 2023.

\bibitem{hodgkin1952quantitative}
A.~L. Hodgkin and A.~F. Huxley, ``A quantitative description of membrane current and its application to conduction and excitation in nerve,'' \emph{The Journal of Physiology}, vol. 117, no.~4, p. 500, 1952.

\end{thebibliography}

\end{document}